\documentclass[journal]{IEEEtran}

\usepackage{amsmath,amsfonts,amssymb,mathtools}
\usepackage{algorithm}
\usepackage{algpseudocode}
\usepackage{array}
\usepackage{tabularx}
\usepackage{booktabs}
\usepackage{bm}
\usepackage{cite}
\usepackage{graphicx}
\usepackage{stfloats}
\usepackage{subfigure}
\usepackage{textcomp}
\usepackage{url}
\usepackage[hidelinks]{hyperref}

\newtheorem{theorem}{Theorem}
\newtheorem{Lemma}{Lemma}
\newtheorem{remark}{Remark}
\newtheorem{corollary}{Corollary}
\newtheorem{proposition}{Proposition}

\DeclareMathOperator*{\argmax}{arg\,max}

\begin{document}

\title{Spectral-NFP: Certified Low-Rank Curvature Majorization for Accelerating WMMSE}

\author{\IEEEauthorblockN{Jianhang Zhu,~\IEEEmembership{Graduate Student Member,~IEEE},
Tsung-Hui Chang,~\IEEEmembership{Fellow,~IEEE},\\ and
Kaiming Shen,~\IEEEmembership{Senior Member,~IEEE}}%
}

\maketitle

\begin{abstract}
Weighted sum-rate maximization in multicell multiple-input multiple-output (MIMO) networks is commonly addressed by the weighted minimum mean-square error (WMMSE) algorithm or fractional programming (FP), both of which, after fixing their auxiliary variables, solve a power-constrained quadratic transmit problem that requires costly dense operations for large arrays. Replacing the underlying curvature matrix with a scaled identity can avoid the matrix inverse operation and thereby reduce complexity, but it discards the curvature eigenvalue structure and yields a loose lower bound. We propose \textbf{Spectral-NFP}, which retains selected dominant curvature eigenpairs and uses a scaled identity matrix to bound the curvature on the remaining subspace from above. The retained rank thus traces a continuous path from NFP to the exact WMMSE transmit update. With the surrogate curvature fixed, Spectral-NFP can be interpreted as Euclidean projected-gradient ascent after a linear coordinate transformation, admitting Nesterov-type acceleration. We derive a lower bound on the one-step transmit-objective gain of Spectral-NFP relative to WMMSE, expressed in terms of the curvature eigenvalues. Under an idealized Wishart model, we analyze this bound in both finite dimensions and the large-system limit, obtaining an asymptotic rank-selection rule. Experimental results show that retaining at most 20\% of the transmit dimension, and often less than 10\%, achieves more than 99\% of the WMMSE WSR. In large-array settings, the measured update time is below 20\% of that required by WMMSE.
\end{abstract}


\section{Introduction}\label{introduction}

\IEEEPARstart{C}{oordinated} multicell multiple-input multiple-output (MIMO) beamforming seeks to maximize the weighted sum rate (WSR) while respecting per-base-station power budgets and suppressing intra-cell and inter-cell interference. Standard weighted minimum mean-square error (WMMSE) and fractional programming (FP) methods address this nonconvex problem through model-driven block updates. In each iteration, their final transmit-beamformer update solves a power-constrained quadratic problem. A direct implementation repeatedly inverts a full-dimensional shifted curvature matrix during the bisection search for the power multiplier, making this update costly for large arrays. One way to reduce this cost is to optimize a quadratic lower bound with a tractable curvature matrix. Nonhomogeneous FP (NFP) adopts this approach by using a scaled identity matrix for the surrogate curvature, which turns the transmit update into a closed-form Euclidean projection and eliminates matrix inversion and bisection. However, replacing the full curvature matrix in this way discards the differences among its eigenvalues and can yield a loose lower bound. We propose \emph{Spectral-NFP} to overcome this limitation. Its surrogate curvature matrix retains selected dominant curvature eigenpairs and uses a scaled identity matrix to bound the curvature on the remaining subspace from above. As the retained rank increases, the curvature bound becomes tighter and the update moves from NFP toward the exact WMMSE transmit update, providing direct control over the accuracy--complexity tradeoff.

Two classical lines of work address linear precoding optimization. The first exploits special objectives or constraint structures to obtain tractable convex or conic formulations. Examples include conic-optimization designs for fixed receivers \cite{WieselEtAl2006Conic}, SINR-constrained downlink beamforming \cite{SchubertBoche2004SINR}, and transmitter optimization under per-antenna constraints \cite{YuLan2007PerAntenna}; their solution structures and multicell extensions have also been studied extensively \cite{BjornsonEtAl2014Beamforming,BjornsonEtAl2010Cooperative,DahroujYu2010Coordinated}. Their convex or conic formulations do not extend directly to general multicell WSR maximization, where every beamformer changes both the desired signals and the interference \cite{Goldsmith2005Wireless,LuoZhang2008DynamicSpectrum}. The second category therefore uses global optimization, such as monotonic optimization and branch-and-bound, to search directly over the resulting nonconvex WSR problem. Although these methods can compute globally optimal solutions for certain formulations, their computational complexity generally limits their practical use to small- or moderate-sized problem instances \cite{JorswieckLarsson2010Monotonic,JoshiEtAl2011BranchBound,LiuEtAl2012GlobalWSR}. Large coordinated arrays consequently rely primarily on scalable stationary-point methods.

WMMSE is the most established model-driven framework for this purpose. The rate--MSE relation was first exploited for MIMO broadcast-channel WSR design \cite{ChristensenEtAl2008WMMSE}, and related MMSE formulations were subsequently developed for the MIMO X channel \cite{AgustinVidal2011MIMOX}, the MIMO interfering broadcast channel \cite{ShiEtAl2011WMMSE}, and cross-layer cellular optimization \cite{BalighEtAl2014WMMSE}. FP provides a complementary auxiliary-variable framework for continuous beamforming and mixed discrete-continuous scheduling problems \cite{ShenYu2018FPPartI,ShenYu2018FPPartII}, with matrix FP extending the same principle to matrix-valued MIMO ratios \cite{ShenEtAl2019MatrixFP}. For the WSR problem considered here, fixing the linear receive filters and auxiliary MSE weight matrices reduces both WMMSE and FP to the same concave quadratic transmit problem. The exact per-cell update involves a curvature matrix whose dimension equals the number of transmit antennas and a power multiplier selected through bisection to satisfy the per-base-station power constraint. In a direct implementation, each bisection step requires the inversion of a shifted full-dimensional curvature matrix. An implementation based on eigenvalue decomposition (EVD) replaces these repeated matrix inversions with a full eigendecomposition of the curvature matrix, thereby avoiding matrix inversion during the bisection search. Nevertheless, it still requires decomposing the full-dimensional curvature matrix.

Several approaches reduce the cost of WSR beamforming by simplifying different parts of the problem. Zero forcing and regularized zero forcing (RZF) prescribe the beamforming structure and avoid iterative WSR optimization \cite{GaoEtAl2011VeryLarge,NguyenEtAl2019RZF}; polynomial and Neumann-series methods approximate the required matrix inverse \cite{KammounEtAl2014Polynomial,PrabhuEtAl2013ApproxInverse}; and coordinated clustering reduces complexity by jointly optimizing only a subset of base stations \cite{HongEtAl2013Clustering}. Manifold optimization instead represents the power-constrained beamformers as Riemannian submanifolds and applies gradient-based updates without large matrix inversions \cite{SunEtAl2024Manifold}. Reduced WMMSE proves that the stationary transmit solution under a sum-power constraint lies in a lower-dimensional subspace generated by the channel matrices and performs the WMMSE updates in that subspace, reducing the antenna-dependent complexity to linear scaling. Its current formulation, however, is restricted to the single-cell setting \cite{ZhaoEtAl2023RethinkingWMMSE}.

The closest model-driven starting point is the nonhomogeneous FP/FastFP construction \cite{ShenEtAl2024AcceleratingQT,ChenEtAl2025FastFPISAC}. After the auxiliary variables are fixed, the transmit block becomes a concave quadratic maximization problem whose negative quadratic term is governed by a positive-semidefinite curvature matrix. Because this objective is maximized, a tractable minorizing surrogate can be constructed by upper-bounding the curvature matrix: the upper bound on the curvature becomes a lower bound on the negative quadratic objective. FastFP chooses the largest eigenvalue of the curvature matrix multiplied by the identity as this upper bound. This choice converts the exact matrix update into an inverse-free residual step followed by a Euclidean projection onto the power constraint. The resulting scalar update also admits a projected-gradient interpretation and can therefore be combined with Nesterov extrapolation \cite{ShenEtAl2024AcceleratingQT,ChenEtAl2025FastFPISAC,Nesterov2018Lectures}. DeepFP follows a data-driven extension of this idea: it unfolds the FastFP iterations and learns more aggressive user-specific scalar stepsizes from channel data instead of fixing every stepsize through the largest curvature eigenvalue \cite{ZhuEtAl2026DeepFP}. Nevertheless, both the model-driven scalar choice in FastFP and the learned scalar choices in DeepFP assign a single curvature value to all transmit directions within each update. Nesterov extrapolation changes the point at which the scalar update is evaluated, while DeepFP changes the scalar stepsize itself; neither restores the directional curvature information discarded by the scaled-identity bound. Consequently, the resulting quadratic surrogate may remain substantially more conservative than the full WMMSE transmit block when the dominant curvature is concentrated in only a few directions.

Spectral-NFP fills this gap by constructing a low-rank-plus-scaled-identity curvature majorizer. It retains a controllable number of leading eigenpairs of each per-cell curvature matrix and bounds all remaining eigenvalues from above by the largest one among them. Retaining more leading eigenpairs progressively tightens the quadratic lower bound and moves the update from NFP toward the exact WMMSE transmit update. By treating only a small number of dominant curvature directions separately, Spectral-NFP obtains a substantially tighter lower bound than NFP without requiring a full-dimensional matrix inverse. The retained rank therefore directly controls the tradeoff between update quality and computational cost, while the per-base-station power constraint is satisfied through a scalar bisection search without matrix inversion inside the search.

Our theory characterizes how the number of retained eigenpairs affects the ratio of the one-step surrogate gain of Spectral-NFP to the one-step gain of the exact WMMSE transmit update. We first derive a deterministic lower bound on this gain ratio in terms of the eigenvalues of the curvature matrix. Under an idealized Wishart curvature model, the classical joint density of the ordered Wishart eigenvalues yields a finite-dimensional lower bound on the probability of attaining a target gain ratio \cite{Muirhead1982Multivariate}. In the large-system limit, the Marchenko--Pastur law and extreme-eigenvalue results further yield an asymptotic lower bound on the gain ratio and an explicit rule for selecting the rank fraction \cite{MarchenkoPastur1967,TulinoVerdu2004RMT,BaiYin1993SmallestEigenvalue}. When the surrogate curvature is fixed, Spectral-NFP also admits a projected-gradient interpretation after a linear coordinate transformation, which motivates applying Nesterov acceleration. In geometry-based experiments with 32 to 512 transmit antennas, Spectral-NFP with retained ranks of 8 and 16 achieves at least 99.35\% and 99.92\% of the WMMSE WSR, respectively. For 256 and 512 transmit antennas, both ranks require less than 20\% of the WMMSE update time.

\begin{figure*}[t]
\centering
\begin{equation*}
f_q(\underline{\bm V},\underline{\bm\Gamma},\underline{\bm Y})
=\sum_{\ell=1}^{L}\sum_{k=1}^{K}
\left[
\operatorname{tr}\left(
2\Re\{\bm V_{\ell k}^H\bm\Lambda_{\ell k}\}
-w_{\ell k}\bm Y_{\ell k}^H\bm D_{\ell k}\bm Y_{\ell k}
(\bm I_d+\bm\Gamma_{\ell k})
\right)
+w_{\ell k}\log|\bm I_d+\bm\Gamma_{\ell k}|
-\operatorname{tr}(w_{\ell k}\bm\Gamma_{\ell k})
\right].
\tag{7}\label{eq:fp-quadratic-objective}
\end{equation*}
\hrulefill
\end{figure*}

The main contributions are summarized as follows.

\begin{enumerate}
\def\labelenumi{\arabic{enumi}.}
\item
  \textbf{Spectral-NFP algorithm.} We propose Spectral-NFP for multicell MIMO WSR maximization. The algorithm retains a selectable number of leading eigenpairs of the WMMSE/FP transmit curvature and bounds the remaining eigenvalues from above by the largest discarded eigenvalue. The number of retained eigenpairs directly controls the progression from NFP toward the exact WMMSE/FP transmit update.
\item
  \textbf{Gain-ratio analysis.} We derive a lower bound on the ratio between the one-step improvement in the quadratic transmit objective achieved by Spectral-NFP and that achieved by the exact WMMSE/FP transmit update. The lower bound is expressed in terms of the eigenvalues of the curvature matrix and increases monotonically as more eigenpairs are retained, showing when a low-rank Spectral-NFP update can approach the gain of WMMSE.
\item
  \textbf{Wishart-model analysis and rank selection.} Under an idealized Wishart model for the curvature matrix, we derive a finite-dimensional probability bound for attaining a target gain ratio. In the large-system limit, the random gain-ratio lower bound converges to a deterministic expression characterized by the limiting eigenvalue distribution. We then use this deterministic expression to derive an explicit criterion for selecting the retained rank fraction.
\end{enumerate}

The rest of this paper is organized as follows. Section II introduces the multicell MIMO model and WSR problem. Section III derives the common WMMSE transmit subproblem and reviews NFP. Section IV develops Spectral-NFP, establishes its projected-gradient interpretation after a linear coordinate transformation, and analyzes its computational complexity. Section V presents the rank-dependent performance analysis. Section VI reports the numerical results, and Section VII concludes the paper.

Here and throughout, bold lower-case letters represent vectors while bold upper-case letters represent matrices. For a vector $\bm{a}$, $\bm{a}^H$ is its conjugate transpose. For a matrix $\bm{A}$, $\bm{A}^H$ is its conjugate transpose and $\|\bm{A}\|_F$ is its Frobenius norm.  $\text{col}(\bm{A})$ refers to the number of columns in matrix $A$. For a square matrix $\bm{A}$, $\text{tr}(\bm{A})$ is its trace, $|\bm{A}|$ is its determinant, and $\lambda_{\max}(\bm{A})$ is its largest eigenvalue. Denote by $\bm{I}_d$ the $d \times d$ identity matrix, $\mathbb{C}^n$ the set of $n \times 1$ vectors, $\mathbb{C}^{d \times n}$ the set of $d \times n$ matrices, and $\mathbb{H}_+^{d \times d}$ the set of $d \times d$ positive definite matrices. For a complex number $\bm{a} \in \mathbb{C}$, $\Re\{\bm{a}\}$ is its real part, $|\bm{a}|$ is its absolute value. The underlined letters represent the collections of the associated vectors or matrices, e.g., for $\bm{a}_1, \dots, \bm{a}_n \in \mathbb{C}^d$ we write $\underline{\bm{a}}  = [\bm{a}_1, \bm{a}_2, \dots, \bm{a}_n]^\top \in \mathbb{C}^{n\times d}$.

\section{Weighted Sum-Rate Maximization Problem}\label{system-model-and-problem-formulation}

Consider a downlink multi-user multiple-input-multiple-output (MU-MIMO) system with $L$ cells. Within each cell, one base station (BS) with $N_t$ transmit antennas serves $K$ users. The $k$th user in the $\ell$th cell is indexed as $(\ell,k)$. Each user $(\ell,k)$ has $N_r$ receive antennas, and $d$ data streams are intended for it. Let $\bm V_{\ell k}\in\mathbb C^{N_t\times d}$ represent the beamforming matrix used by BS $\ell$ associated with the signal $\bm s_{\ell k}\in\mathbb C^{d\times1}$ for user $(\ell,k)$. Assuming that $\mathbb E[\bm s_{\ell k}\bm s_{\ell k}^H]=\bm I_d$, the received signal $\bm y_{\ell k}$ at user $(\ell,k)$ is given by
\begin{equation}
\begin{aligned}
\bm y_{\ell k}
={}&\underbrace{\bm H_{\ell k,\ell}\bm V_{\ell k}\bm s_{\ell k}}_{\text{desired signal}}
+\underbrace{\sum_{j=1,j\ne k}^{K}\bm H_{\ell k,\ell}\bm V_{\ell j}\bm s_{\ell j}}_{\text{intracell interference}}\\
&+\underbrace{\sum_{i=1,i\ne\ell}^{L}\sum_{j=1}^{K}\bm H_{\ell k,i}\bm V_{ij}\bm s_{ij}}_{\text{intercell interference}}
+\bm n_{\ell k},
\end{aligned}
\end{equation}
where the channel state information (CSI) $\bm H_{\ell k,i}\in\mathbb C^{N_r\times N_t}$ is the channel from BS $i$ to user $(\ell,k)$, and $\bm n_{\ell k}\sim\mathcal{CN}(\mathbf0,\sigma^2\bm I_{N_r})$ is the additive white Gaussian noise with variance $\sigma^2$. The achievable data rate for user $(\ell,k)$ can be computed as~\cite{Goldsmith2005Wireless}
\begin{equation}
R_{\ell k}
=\log\left|\bm I_d
+\bm V_{\ell k}^H\bm H_{\ell k,\ell}^H
\bm F_{\ell k}^{-1}
\bm H_{\ell k,\ell}\bm V_{\ell k}\right|,
\end{equation}
where
\begin{multline}
\bm F_{\ell k}
=\sum_{j=1,j\ne k}^{K}
\bm H_{\ell k,\ell}\bm V_{\ell j}\bm V_{\ell j}^H
\bm H_{\ell k,\ell}^H\\
+\sum_{i=1,i\ne\ell}^{L}\sum_{j=1}^{K}
\bm H_{\ell k,i}\bm V_{ij}\bm V_{ij}^H
\bm H_{\ell k,i}^H
+\sigma^2\bm I_{N_r}.
\label{eq:interference-covariance}
\end{multline}

We seek the optimal transmit beamformers $\underline{\bm V}$ to maximize the weighted sum rates:
\begin{equation}
\begin{aligned}
\max_{\underline{\bm V}}\quad
&f_o(\underline{\bm V})
:=\sum_{\ell=1}^{L}\sum_{k=1}^{K}w_{\ell k}R_{\ell k}\\
\text{s.t.}\quad
&\sum_{k=1}^{K}\operatorname{tr}
(\bm V_{\ell k}\bm V_{\ell k}^H)
\le P_\ell,
\quad \ell=1,2,\ldots,L,
\end{aligned}
\label{eq:wsr-problem}
\end{equation}
where the nonnegative weight $w_{\ell k}\ge0$ reflects the priority of user $(\ell,k)$, and the constant $P_\ell$ is the power budget of BS $\ell$.

Denote the power-feasible beamformer set by
\begin{equation*}
\mathcal C
\triangleq
\left\{
\underline{\bm V}:
\sum_{k=1}^{K}\|\bm V_{\ell k}\|_F^2\le P_\ell,
\quad \ell=1,\ldots,L
\right\}.
\end{equation*}

The principal matrices and variables are summarized in Table~\ref{tab:dimensions}.

\begin{table}[t]
\renewcommand{\arraystretch}{1.5}
\footnotesize
\centering
\caption{Principal matrices and variables.}
\label{tab:dimensions}
\begin{tabularx}{\columnwidth}{|>{\centering\arraybackslash}p{0.25\columnwidth}||>{\raggedright\arraybackslash}X|}
\hline
\textbf{Variable} & \textbf{Meaning} \\ \hline
\hline
$\bm H_{\ell k,i}$ & Channel from BS $i$ to user $(\ell,k)$ \\ \hline
$\bm V_{\ell k}$ & Transmit beamforming matrix for user $(\ell,k)$ \\ \hline
$\bm\Lambda_{\ell k}$ & Linear-term coefficient in the FP transmit objective \\ \hline
$\bm B_{\ell k}^{(r)}$ & Linear-term coefficient in the Spectral-NFP surrogate \\ \hline
$\bm F_{\ell k}$ & Interference-plus-noise covariance matrix \\ \hline
$\bm D_{\ell k}$ & Total receive covariance matrix \\ \hline
$\bm Y_{\ell k}$ & Auxiliary variable introduced by the quadratic transform \\ \hline
$\bm\Gamma_{\ell k}$ & Auxiliary variable introduced by the Lagrangian dual transform \\ \hline
$\bm Z_{\ell k}$ & Auxiliary variable introduced by the nonhomogeneous bound \\ \hline
$\bm L_\ell$ & Curvature matrix of the FP transmit objective \\ \hline
$\bm K_{\ell,r}$ & Spectral curvature upper bound for $\bm L_\ell$ \\ \hline
$\bm U_{\ell,r}$ & Retained curvature eigenvectors \\ \hline
$\eta_\ell$ & Cell-wise power multiplier \\ \hline
\end{tabularx}
\end{table}

\section{WMMSE Transmit Subproblem and NFP}\label{fpwmmse-transmit-subproblem-and-nfp}

\subsection{FP and WMMSE Transmit Update}\label{quadratic-transmit-subproblem}

\begin{figure*}[!b]
\centering
\hrulefill
\begin{multline*}
f_n(\underline{\bm V},\underline{\bm\Gamma},\underline{\bm Y},\underline{\bm Z})
=\sum_{\ell=1}^{L}\sum_{k=1}^{K}\Big[
\operatorname{tr}\Big(
2\operatorname{Re}\big\{
\bm V_{\ell k}^H\bm\Lambda_{\ell k}
+\bm V_{\ell k}^H(\lambda_\ell\bm I_{N_t}-\bm L_\ell)\bm Z_{\ell k}
\big\}
+\bm Z_{\ell k}^H(\bm L_\ell-\lambda_\ell\bm I_{N_t})\bm Z_{\ell k}
-\lambda_\ell\bm V_{\ell k}^H\bm V_{\ell k}
\Big)\\
-\operatorname{tr}\Big(
w_{\ell k}\sigma^2(\bm I_d+\bm\Gamma_{\ell k})
\bm Y_{\ell k}^H\bm Y_{\ell k}
\Big)
+w_{\ell k}\log|\bm I_d+\bm\Gamma_{\ell k}|
-\operatorname{tr}(w_{\ell k}\bm\Gamma_{\ell k})
\Big].
\tag{17}\label{eq:nfp-objective}
\end{multline*}
\end{figure*}

By the Lagrangian dual transform~\cite{ShenYu2018FPPartI}, the original objective $f_o(\underline{\bm V})$ is converted to
\begin{align}
f_r(\underline{\bm V},\underline{\bm\Gamma})
&=\sum_{\ell=1}^{L}\sum_{k=1}^{K}w_{\ell k}
\Big[\log|\bm I_d+\bm\Gamma_{\ell k}|
-\text{tr}(\bm\Gamma_{\ell k})\notag\\
&\quad+\text{tr}\big((\bm I_d+\bm\Gamma_{\ell k})
\bm V_{\ell k}^H\bm H_{\ell k,\ell}^H
\bm D_{\ell k}^{-1}\bm H_{\ell k,\ell}
\bm V_{\ell k}\big)\Big],
\end{align}
where
\begin{equation}
\bm D_{\ell k}
=\sum_{i=1}^{L}\sum_{j=1}^{K}
\bm H_{\ell k,i}\bm V_{ij}\bm V_{ij}^H
\bm H_{\ell k,i}^H
+\sigma^2\bm I_{N_r}.
\label{eq:total-receive-covariance}
\end{equation}
The FP method then applies the quadratic transform~\cite{ShenYu2018FPPartI,ShenEtAl2019MatrixFP} to further recast $f_o(\underline{\bm V})$ into the objective $f_q(\underline{\bm V},\underline{\bm\Gamma},\underline{\bm Y})$ given in \eqref{eq:fp-quadratic-objective}. The new objective is separately concave in $\underline{\bm V}$, $\underline{\bm\Gamma}$, and $\underline{\bm Y}$, so the FP algorithm allows iteratively optimizing these variables as
\setcounter{equation}{7}
\begin{align}
\bm Y_{\ell k}
&=\bm D_{\ell k}^{-1}
\bm H_{\ell k,\ell}\bm V_{\ell k},
\label{eq:fp-y-update}\\
\bm\Gamma_{\ell k}
&=\bm V_{\ell k}^H
\bm H_{\ell k,\ell}^H
\bm F_{\ell k}^{-1}
\bm H_{\ell k,\ell}\bm V_{\ell k}.
\label{eq:fp-gamma-update}
\end{align}
With $\underline{\bm\Gamma}$ and $\underline{\bm Y}$ fixed, the terms in $f_q$ that depend on the transmit beamformers are
\begin{equation}
Q(\underline{\bm V})
=\sum_{\ell=1}^{L}\sum_{k=1}^{K}
\left[
2\operatorname{Re}\operatorname{tr}
(\bm V_{\ell k}^H\bm\Lambda_{\ell k})
-\operatorname{tr}
(\bm V_{\ell k}^H\bm L_\ell\bm V_{\ell k})
\right],
\end{equation}
where
\begin{align}
\bm\Lambda_{\ell k}
&=w_{\ell k}
\bm H_{\ell k,\ell}^H
\bm Y_{\ell k}(\bm I_d+\bm\Gamma_{\ell k}),
\label{eq:transmit-linear-term}\\
\bm L_\ell
&=\sum_{i=1}^{L}\sum_{j=1}^{K}
w_{ij}
\bm H_{ij,\ell}^H
\bm Y_{ij}(\bm I_d+\bm\Gamma_{ij})\bm Y_{ij}^H
\bm H_{ij,\ell}.
\label{eq:transmit-curvature}
\end{align}
With the auxiliary variables fixed, the transmit beamformers are obtained by solving
\begin{equation}
\begin{aligned}
\max_{\underline{\bm V}}
\quad&
\sum_{\ell=1}^{L}\sum_{k=1}^{K}
\left[
2\operatorname{Re}\operatorname{tr}
(\bm V_{\ell k}^H\bm\Lambda_{\ell k})
-\operatorname{tr}
(\bm V_{\ell k}^H\bm L_\ell\bm V_{\ell k})
\right]\\
\text{s.t.}\quad&
\sum_{k=1}^{K}\|\bm V_{\ell k}\|_F^2\le P_\ell,
\quad \ell=1,\ldots,L.
\end{aligned}
\label{eq:quadratic-transmit-problem}
\end{equation}
Since the power constraints are imposed separately at the BSs, the optimal transmit update is
\begin{equation}
\bm V_{\ell k}
=(\eta_\ell\bm I_{N_t}+\bm L_\ell)^{-1}
\bm\Lambda_{\ell k},
\end{equation}
where the Lagrange multiplier $\eta_\ell$ for the power constraint is computed as
\begin{equation*}
\eta_\ell
=\min\left\{\eta\ge0:
\sum_{k=1}^{K}
\left\|(\eta\bm I_{N_t}+\bm L_\ell)^{-1}
\bm\Lambda_{\ell k}\right\|_F^2
\le P_\ell\right\}.
\end{equation*}
When $\eta_\ell=0$ and $\bm L_\ell$ is singular, the inverse in the exact
update is interpreted as the Moore--Penrose pseudoinverse, which gives the
minimum-norm solution on the range of $\bm L_\ell$.
The WMMSE algorithm~\cite{ShiEtAl2011WMMSE,ChristensenEtAl2008WMMSE} gives the same auxiliary-variable and transmit updates for the WSR problem considered here.

\subsection{NFP as Isotropic Curvature Majorization}\label{nfp-as-isotropic-curvature-majorization}

The exact FP/WMMSE transmit update in (14) requires an inverse of the $N_t\times N_t$ curvature matrix and a bisection search for the power multiplier. NFP~\cite{ShenEtAl2024AcceleratingQT,ChenEtAl2025FastFPISAC} reduces this cost by replacing the original quadratic transmit objective with a quadratic lower-bound surrogate whose curvature matrix is a scaled identity matrix. This construction rests on the nonhomogeneous bound below.

\begin{Lemma}[Nonhomogeneous Bound~\cite{SunBabuPalomar2017MM}]
\label{lem:nonhomogeneous-bound}
Suppose that two Hermitian matrices \(\bm L,\bm K\in\mathbb H^{m\times m}\) satisfy \(\bm L\preceq\bm K\). Then, for any two compatible matrices \(\bm X\) and \(\bm Z\), one has
\begin{multline}
\operatorname{tr}(\bm X^H\bm L\bm X)
\le
\operatorname{tr}\left(
\bm X^H\bm K\bm X
+2\operatorname{Re}\{\bm X^H(\bm L-\bm K)\bm Z\}
\right.\\
\left.
+\bm Z^H(\bm K-\bm L)\bm Z
\right),
\label{eq:nonhomogeneous-bound}
\end{multline}
where equality holds if \(\bm Z=\bm X\).
\end{Lemma}

To simplify the quadratic transmit problem \eqref{eq:quadratic-transmit-problem}, NFP applies Lemma~\ref{lem:nonhomogeneous-bound} to each quadratic term \(\operatorname{tr}(\bm V_{\ell k}^{H}\bm L_\ell\bm V_{\ell k})\), with \(\bm X=\bm V_{\ell k}\), \(\bm Z=\bm Z_{\ell k}\), \(\bm L=\bm L_\ell\), and \(\bm K=\lambda_\ell\bm I_{N_t}\), where
\begin{equation}
\lambda_\ell=\lambda_{\max}(\bm L_\ell).
\label{eq:nfp-lambda}
\end{equation}
Substituting the resulting lower bounds into \(f_q\) gives \(f_n\) in \eqref{eq:nfp-objective}. When the other variables are held fixed, the bound is tight at
\setcounter{equation}{17}
\begin{equation}
\bm Z_{\ell k}=\bm V_{\ell k}.
\label{eq:nfp-z-update}
\end{equation}
Because the surrogate curvature is a scaled identity matrix, maximizing \eqref{eq:nfp-objective} under the per-BS power constraint reduces to a residual step followed by a Euclidean projection onto the power ball. Thus, each \(\bm V_{\ell k}\) at iteration \(t\) is updated from the previous iterate as
\begin{equation}
\bm V_{\ell k}^{(t)}=
\begin{cases}
\widehat{\bm V}_{\ell k},
&\text{if }\displaystyle\sum_{j=1}^{K}
\|\widehat{\bm V}_{\ell j}\|_F^2\le P_\ell,\\[1.2ex]
\displaystyle
\sqrt{\frac{P_\ell}{\sum_{j=1}^{K}
\|\widehat{\bm V}_{\ell j}\|_F^2}}
\widehat{\bm V}_{\ell k},
&\text{otherwise},
\end{cases}
\label{eq:nfp-v-update}
\end{equation}
where
\begin{equation}
\widehat{\bm V}_{\ell k}
=\bm V_{\ell k}^{(t-1)}
+\frac{1}{\lambda_\ell}
\left(
\bm\Lambda_{\ell k}
-\bm L_\ell\bm V_{\ell k}^{(t-1)}
\right).
\label{eq:nfp-vhat}
\end{equation}
Algorithm~\ref{alg:nfp} summarizes the NFP method.

\begin{algorithm}[t]
\caption{NFP for Multicell MIMO Beamforming}
\label{alg:nfp}
\begin{algorithmic}[1]
\State \textbf{Input:} The current CSI.
\State Initialize \(\underline{\bm V}\) to feasible values under the power constraint.
\Repeat
  \State Update each \(\bm Z_{\ell k}\) by \eqref{eq:nfp-z-update}.
  \State Update each \(\bm Y_{\ell k}\) by \eqref{eq:fp-y-update}.
  \State Update each \(\bm\Gamma_{\ell k}\) by \eqref{eq:fp-gamma-update}.
  \State Update each \(\bm V_{\ell k}\) by \eqref{eq:nfp-v-update}.
\Until{the objective value converges.}
\State \textbf{Output:} Final beamforming matrix \(\underline{\bm V}\).
\end{algorithmic}
\end{algorithm}

\begin{remark}
\label{remark:nfp-scaled-identity}
The computational simplicity of NFP comes from replacing every eigenvalue of \(\bm L_\ell\) by the single value \(\lambda_{\max}(\bm L_\ell)\). This replacement removes the differences among the eigenvalues of the full curvature matrix. Consequently, when only a few eigenvalues are close to the largest one, the NFP lower bound can be unnecessarily loose in the remaining eigen-directions. Spectral-NFP overcomes this limitation by preserving selected dominant eigenpairs and applying a scaled-identity bound only to the residual subspace.
\end{remark}

\section{Proposed Spectral-NFP}\label{proposed-spectral-nfp}

NFP simplifies the transmit update by replacing the curvature matrix \(\bm L_\ell\) with the scaled identity \(\lambda_{\ell,1}\bm I_{N_t}\). Although this choice eliminates matrix inversion and bisection, it assigns the largest eigenvalue to every eigendirection and may therefore produce a loose quadratic lower bound. This motivates seeking a curvature upper bound that is closer to \(\bm L_\ell\) while retaining a tractable power-constrained update. Spectral-NFP achieves this objective by preserving a selected number of dominant eigenpairs of \(\bm L_\ell\) and upper-bounding only the remaining eigenvalues.

\subsection{\texorpdfstring{Rank-\(r\) Spectral Majorizer}{Rank-r Spectral Majorizer}}\label{rank-r-spectral-majorizer}

The retained rank can take any value in $r\in\{0,\ldots,N_t-1\}$. Consider the eigendecomposition
\begin{equation}
\bm L_\ell
=\bm U_\ell
\operatorname{diag}
(\lambda_{\ell,1},\ldots,\lambda_{\ell,N_t})
\bm U_\ell^H,
\end{equation}
where \(\lambda_{\ell,1}\ge\lambda_{\ell,2}\ge\cdots\ge
\lambda_{\ell,N_t}\ge0\). By definition,
\(\lambda_{\ell,1}=\lambda_{\max}(\bm L_\ell)=\lambda_\ell\), where
\(\lambda_\ell\) is the scalar curvature used by NFP. For
\(0<r\le N_t-1\), define
\begin{equation}
\bm U_{\ell,r}
=[\bm u_{\ell,1},\ldots,\bm u_{\ell,r}],
\qquad
\bm\Sigma_{\ell,r}
=\operatorname{diag}(\lambda_{\ell,1},\ldots,\lambda_{\ell,r}),
\end{equation}
For every admissible rank, the largest discarded eigenvalue is
\begin{equation}
\beta_{\ell,r}=\lambda_{\ell,r+1}.
\label{eq:spectral-beta}
\end{equation}

For \(r=0\), the retained eigenvector and eigenvalue blocks are absent, and
\(\beta_{\ell,0}=\lambda_{\ell,1}\). Spectral-NFP introduces the following
curvature upper bound:
\begin{equation}
\bm K_{\ell,r}
=\bm U_\ell
\operatorname{diag}\Bigl(
\lambda_{\ell,1},\ldots,\lambda_{\ell,r},
\underbrace{\beta_{\ell,r},\ldots,\beta_{\ell,r}}_{N_t-r}
\Bigr)
\bm U_\ell^H.
\label{eq:spectral-majorizer}
\end{equation}
Thus, \(\bm K_{\ell,r}\) has the same eigenvectors as
\(\bm L_\ell\). It preserves the leading \(r\) eigenvalues
\(\lambda_{\ell,1},\ldots,\lambda_{\ell,r}\) and replaces every remaining
eigenvalue by the common residual value \(\beta_{\ell,r}\). At \(r=0\), this
definition gives \(\bm K_{\ell,0}=\lambda_{\ell,1}\bm I_{N_t}\). At
\(r=N_t-1\), the residual subspace contains only the last eigenvector and is
assigned \(\beta_{\ell,N_t-1}=\lambda_{\ell,N_t}\); hence
\(\bm K_{\ell,N_t-1}=\bm L_\ell\).

\begin{proposition}[Spectral Majorization]
\label{prop:spectral-majorization}
For \(0\le r<N_t-1\),
\begin{equation}
\bm L_\ell=\bm K_{\ell,N_t-1}
\preceq\bm K_{\ell,r+1}
\preceq\bm K_{\ell,r}
\preceq\bm K_{\ell,0}
=\lambda_{\ell,1}\bm I_{N_t}.
\label{eq:spectral-majorization}
\end{equation}

\end{proposition}

\begin{IEEEproof}
All matrices in \eqref{eq:spectral-majorization} are diagonal in the eigenbasis \(\bm U_\ell\). For rank \(r\), the first \(r\) diagonal entries of \(\bm K_{\ell,r}\) are \(\lambda_{\ell,1},\ldots,\lambda_{\ell,r}\), while every remaining entry is \(\lambda_{\ell,r+1}\). Increasing the rank to \(r+1\) preserves the first \(r+1\) eigenvalues and replaces the residual value \(\lambda_{\ell,r+1}\) by the no-larger value \(\lambda_{\ell,r+2}\). Hence \(\bm K_{\ell,r+1}\preceq\bm K_{\ell,r}\). Moreover, each residual eigenvalue of \(\bm L_\ell\) satisfies \(\lambda_{\ell,a}\le\lambda_{\ell,r+2}\) for \(a>r+1\), so \(\bm L_\ell\preceq\bm K_{\ell,r+1}\). Finally, every eigenvalue of \(\bm K_{\ell,r}\) is at most \(\lambda_{\ell,1}\). The endpoint identities follow directly from \(\bm K_{\ell,0}=\lambda_{\ell,1}\bm I_{N_t}\) and \(\beta_{\ell,N_t-1}=\lambda_{\ell,N_t}\), which gives \(\bm K_{\ell,N_t-1}=\bm L_\ell\).
\end{IEEEproof}

For \(0<r\le N_t-1\), partition
\(\bm U_\ell=[\bm U_{\ell,r}\ \bm U_{\ell,r}^{\perp}]\).
By the completeness relation of the orthonormal eigenbasis,
\(\bm U_{\ell,r}\bm U_{\ell,r}^H
+\bm U_{\ell,r}^{\perp}(\bm U_{\ell,r}^{\perp})^H
=\bm I_{N_t}\), or equivalently,
\(\bm U_{\ell,r}^{\perp}(\bm U_{\ell,r}^{\perp})^H
=\bm I_{N_t}-\bm U_{\ell,r}\bm U_{\ell,r}^H\),
\eqref{eq:spectral-majorizer} becomes
\begin{equation}
\bm K_{\ell,r}
=\beta_{\ell,r}\bm I_{N_t}
+\bm U_{\ell,r}
(\bm\Sigma_{\ell,r}-\beta_{\ell,r}\bm I_r)
\bm U_{\ell,r}^H.
\label{eq:spectral-majorizer-low-rank}
\end{equation}

\subsection{Spectral-NFP Transmit Update}\label{spectral-nfp-transmit-update}

Applying \eqref{eq:nonhomogeneous-bound} to the quadratic terms of
\(f_q\), with \(\bm K=\bm K_{\ell,r}\) and
\(\bm Z_{\ell k}=\bm V_{\ell k}^{(t)}\), gives the
Spectral-NFP surrogate objective
\begin{equation}
\begin{aligned}
f_{s,r}\!\left(
\underline{\bm V}\mid\underline{\bm V}^{(t)}
\right)
\triangleq{}&
\sum_{\ell=1}^{L}\sum_{k=1}^{K}
2\operatorname{Re}\operatorname{tr}
(\bm V_{\ell k}^H\bm B_{\ell k}^{(r)})\\
&-
\sum_{\ell=1}^{L}\sum_{k=1}^{K}
\operatorname{tr}
(\bm V_{\ell k}^H\bm K_{\ell,r}\bm V_{\ell k})
+\text{const}.
\end{aligned}
\label{eq:spectral-surrogate-objective}
\end{equation}
where
\begin{equation}
\bm B_{\ell k}^{(r)}
\triangleq\bm\Lambda_{\ell k}
+(\bm K_{\ell,r}-\bm L_\ell)
\bm V_{\ell k}^{(t)}.
\label{eq:spectral-linear-term}
\end{equation}

Since
\(\bm L_\ell\preceq\bm K_{\ell,r}
\preceq\lambda_{\ell,1}\bm I_{N_t}\), the nonhomogeneous
bound gives, with the auxiliary variables fixed at iteration \(t\),
\begin{equation}
\begin{aligned}
&f_n\!\left(
\underline{\bm V},
\underline{\bm\Gamma}^{(t)},
\underline{\bm Y}^{(t)},
\underline{\bm V}^{(t)}
\right)\\
&\qquad\le
f_{s,r}\!\left(
\underline{\bm V}\mid\underline{\bm V}^{(t)}
\right)
\le
f_q\!\left(
\underline{\bm V},
\underline{\bm\Gamma}^{(t)},
\underline{\bm Y}^{(t)}
\right).
\end{aligned}
\label{eq:nfp-spectral-fq-ordering}
\end{equation}
All three functions are equal at
\(\underline{\bm V}=\underline{\bm V}^{(t)}\). Therefore,
Spectral-NFP provides a tighter lower bound than NFP while remaining a
lower bound on \(f_q\), as illustrated in Fig.~\ref{fig:spectral-surrogate-geometry}.

\begin{figure}[t]
\centering
\includegraphics[width=\columnwidth]{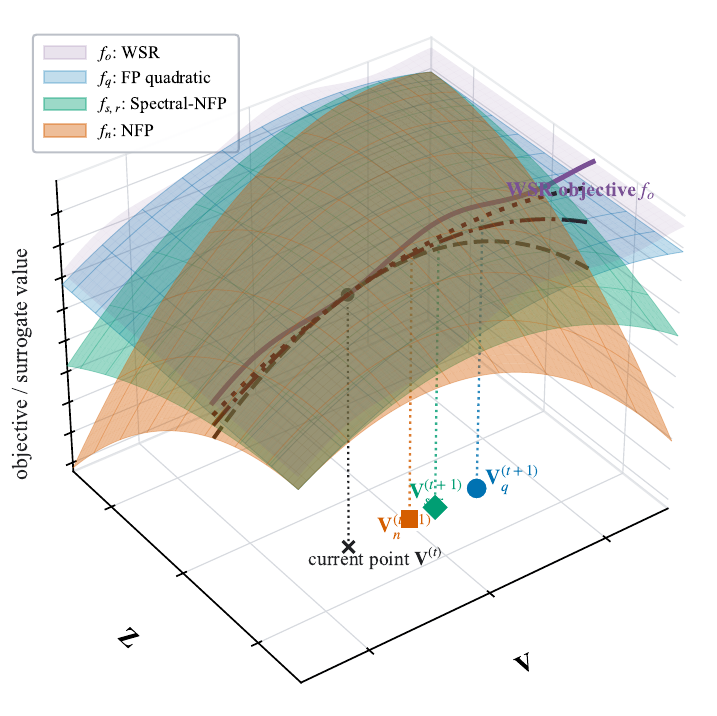}
\caption{Three-dimensional illustration of surrogate tightness. The horizontal
axes are scalar projections of \(\underline{\bm V}\) and
\(\underline{\bm Z}\). With the auxiliary variables fixed, the surfaces
satisfy \(f_n\le f_{s,r}\le f_q\le f_o\) and touch at the current beamformer.
The markers show the corresponding updates on
\(\underline{\bm Z}=\underline{\bm V}^{(t)}\).}
\label{fig:spectral-surrogate-geometry}
\end{figure}

Maximizing \eqref{eq:spectral-surrogate-objective} under the per-BS power
constraints, the KKT condition for BS \(\ell\), with multiplier
\(\eta_\ell\ge0\), is
\begin{equation}
(\bm K_{\ell,r}+\eta_\ell\bm I_{N_t})
\bm V_{\ell k}^{(t+1)}
=\bm B_{\ell k}^{(r)},
\label{eq:spectral-kkt-system}
\end{equation}
and hence
\begin{equation}
\bm V_{\ell k}^{(t+1)}
=(\bm K_{\ell,r}+\eta_\ell\bm I_{N_t})^{-1}
\bm B_{\ell k}^{(r)}.
\label{eq:spectral-kkt-update}
\end{equation}

For \(0<r\le N_t-1\) and \(\beta_{\ell,r}+\eta_\ell>0\), directly inverting
the diagonal entries in \eqref{eq:spectral-majorizer} gives
\begin{multline*}
(\bm K_{\ell,r}+\eta_\ell\bm I_{N_t})^{-1}\\
=
\bm U_\ell
\begin{bmatrix}
(\bm\Sigma_{\ell,r}+\eta_\ell\bm I_r)^{-1}
&\mathbf0\\
\mathbf0
&(\beta_{\ell,r}+\eta_\ell)^{-1}\bm I_{N_t-r}
\end{bmatrix}
\bm U_\ell^H.
\end{multline*}
As in the derivation of \eqref{eq:spectral-majorizer-low-rank}, block-diagonal
multiplication and the completeness relation of the eigenbasis then yield
\begin{equation}
\begin{aligned}
(\bm K_{\ell,r}+\eta_\ell\bm I_{N_t})^{-1}
&=
\bm U_{\ell,r}
(\bm\Sigma_{\ell,r}+\eta_\ell\bm I_r)^{-1}
\bm U_{\ell,r}^H\\
&\quad+
\frac{1}{\beta_{\ell,r}+\eta_\ell}
(\bm I_{N_t}-\bm U_{\ell,r}\bm U_{\ell,r}^H).
\end{aligned}
\label{eq:spectral-low-rank-inverse}
\end{equation}

Substituting \eqref{eq:spectral-low-rank-inverse} into \eqref{eq:spectral-kkt-update} yields
\begin{align}
&\bm V_{\ell k}^{(t+1)}
=
\frac{\bm B_{\ell k}^{(r)}}
{\beta_{\ell,r}+\eta_\ell}\,+\notag\\
&\;\bm U_{\ell,r}
\operatorname{diag}\!\left\{
(\lambda_{\ell,a}+\eta_\ell)^{-1}
-(\beta_{\ell,r}+\eta_\ell)^{-1}
\right\}_{a=1}^{r}
\bm U_{\ell,r}^H\bm B_{\ell k}^{(r)}.
\label{eq:spectral-low-rank-update}
\end{align}

After the eigenpairs have been obtained, \eqref{eq:spectral-low-rank-update} uses only projections onto the retained \(r\)-dimensional subspace. At \(r=0\), \eqref{eq:spectral-surrogate-objective} is the NFP surrogate and its optimizer is \eqref{eq:nfp-v-update}. At \(r=N_t-1\), \(\bm K_{\ell,N_t-1}=\bm L_\ell\) and \(\bm B_{\ell k}^{(N_t-1)}=\bm\Lambda_{\ell k}\), so \eqref{eq:spectral-kkt-update} becomes the exact FP/WMMSE transmit update.

From \eqref{eq:spectral-low-rank-update}, the total transmit power of BS
\(\ell\) is the scalar function
\begin{multline}
p_\ell(\eta)
=
\sum_{k=1}^{K}
\Biggl[
\sum_{a=1}^{r}
\frac{\|\bm u_{\ell,a}^H\bm B_{\ell k}^{(r)}\|_2^2}
{(\lambda_{\ell,a}+\eta)^2}\\
+
\frac{\|(\bm I_{N_t}-\bm U_{\ell,r}\bm U_{\ell,r}^H)
\bm B_{\ell k}^{(r)}\|_F^2}
{(\beta_{\ell,r}+\eta)^2}
\Biggr].
\label{eq:spectral-power-equation}
\end{multline}
It is continuous and nonincreasing for \(\eta\ge0\). Hence,
\(\eta_\ell=0\) if \(p_\ell(0)\le P_\ell\); otherwise,
\(\eta_\ell\) is determined from \(p_\ell(\eta_\ell)=P_\ell\) by scalar
bisection.
When \(\eta_\ell=0\) and \(\bm K_{\ell,r}\) is singular, the spectral
inverse and the zero-eigenvalue terms in \(p_\ell(0)\) are interpreted in the
same Moore--Penrose sense; if \(\lambda_\ell=0\) in NFP, the current feasible
beamformers are retained.

\begin{algorithm}[t]
\caption{Spectral-NFP}
\label{alg:spectral-nfp}
\begin{algorithmic}[1]
\State \textbf{Input:} The current CSI and retained rank
\(r\in\{1,\ldots,N_t-1\}\).
\State Initialize \(\underline{\bm V}\) to feasible values under the power constraint.
\Repeat
  \State Update each \(\bm Z_{\ell k}\) by \eqref{eq:nfp-z-update}.
  \State Update each \(\bm Y_{\ell k}\) by \eqref{eq:fp-y-update}.
  \State Update each \(\bm\Gamma_{\ell k}\) by \eqref{eq:fp-gamma-update}.
  \State Construct each \(\bm K_{\ell,r}\) by \eqref{eq:spectral-majorizer}.
  \State Update each \(\bm B_{\ell k}^{(r)}\) by \eqref{eq:spectral-linear-term}.
  \State Update each \(\bm V_{\ell k}\) by
  \eqref{eq:spectral-low-rank-update}.
\Until{the objective value converges.}
\State \textbf{Output:} Final beamforming matrix \(\underline{\bm V}\).
\end{algorithmic}
\end{algorithm}

\subsection{Projected Gradient in Transformed Coordinates}\label{transformed-coordinate-projected-gradient-view}

We next show that, after a linear coordinate transformation, the Spectral-NFP update is exactly a projected-gradient-ascent step on the transmit-beamformer block of \(f_q\). With the auxiliary variables and the spectral majorizers fixed, assume \(\bm K_{\ell,r}\succ\mathbf0\) for every cell and define the bijective transformation
\begin{equation*}
\widetilde{\bm V}_{\ell k}
=\bm K_{\ell,r}^{1/2}\bm V_{\ell k},
\qquad
\bm V_{\ell k}
=\bm K_{\ell,r}^{-1/2}\widetilde{\bm V}_{\ell k}.
\end{equation*}
The original power constraints become
\begin{equation*}
\widetilde{\mathcal C}_{r}
\triangleq
\left\{\underline{\widetilde{\bm V}}:
\sum_{k=1}^{K}\operatorname{tr}\left(
\widetilde{\bm V}_{\ell k}^H
\bm K_{\ell,r}^{-1}
\widetilde{\bm V}_{\ell k}\right)
\le P_\ell,
\quad \ell=1,\ldots,L
\right\}.
\end{equation*}
Define the transformed quadratic transmit objective as
\begin{equation*}
\begin{aligned}
\widetilde f_q(\underline{\widetilde{\bm V}})
={}&\sum_{\ell=1}^{L}\sum_{k=1}^{K}
2\operatorname{Re}\operatorname{tr}\left(
\widetilde{\bm V}_{\ell k}^{H}
\bm K_{\ell,r}^{-1/2}\bm\Lambda_{\ell k}
\right)\\
&-\sum_{\ell=1}^{L}\sum_{k=1}^{K}
\operatorname{tr}\left(
\widetilde{\bm V}_{\ell k}^{H}
\bm K_{\ell,r}^{-1/2}\bm L_\ell
\bm K_{\ell,r}^{-1/2}
\widetilde{\bm V}_{\ell k}
\right)
+\text{const}.
\end{aligned}
\end{equation*}

Differentiating \(\widetilde f_q\) gives
\begin{align*}
\left.
\frac{\partial \widetilde f_q}
{\partial\widetilde{\bm V}_{\ell k}^{*}}
\right|_{\underline{\widetilde{\bm V}}=\underline{\widetilde{\bm V}}^{(t)}}
&=\bm K_{\ell,r}^{-1/2}
\left(
\bm\Lambda_{\ell k}
-\bm L_\ell\bm V_{\ell k}^{(t)}
\right).
\end{align*}
Therefore, the projected-gradient update is
\begin{equation*}
\begin{aligned}
\widetilde{\bm V}_{\ell k}^{(t+1)}
&=\mathcal P_{\widetilde{\mathcal C}_{r}}
\left(
\widetilde{\bm V}_{\ell k}^{(t)}
+\left.
\frac{\partial \widetilde f_q}
{\partial\widetilde{\bm V}_{\ell k}^{*}}
\right|_{\underline{\widetilde{\bm V}}=\underline{\widetilde{\bm V}}^{(t)}}
\right)\\
&=\mathcal P_{\widetilde{\mathcal C}_{r}}
\left(
\bm K_{\ell,r}^{-1/2}
\left(
\bm K_{\ell,r}\bm V_{\ell k}^{(t)}
+\bm\Lambda_{\ell k}
-\bm L_\ell\bm V_{\ell k}^{(t)}
\right)
\right)\\
&=\mathcal P_{\widetilde{\mathcal C}_{r}}
\left(
\bm K_{\ell,r}^{-1/2}\bm B_{\ell k}^{(r)}
\right),
\end{aligned}
\end{equation*}
where \(\mathcal P_{\widetilde{\mathcal C}_{r}}(\cdot)\) denotes Euclidean projection onto \(\widetilde{\mathcal C}_{r}\), applied jointly to the beamformers subject to each cell power constraint. The second equality substitutes the gradient above, and the third follows from \eqref{eq:spectral-linear-term}. The projection separates across cells, and its KKT condition for cell \(\ell\) is
\begin{equation*}
\left(\bm I_{N_t}+\eta_\ell\bm K_{\ell,r}^{-1}\right)
\widetilde{\bm V}_{\ell k}^{(t+1)}
=\bm K_{\ell,r}^{-1/2}\bm B_{\ell k}^{(r)}.
\end{equation*}
Substituting \(\widetilde{\bm V}_{\ell k}^{(t+1)}=\bm K_{\ell,r}^{1/2}\bm V_{\ell k}^{(t+1)}\) and multiplying by \(\bm K_{\ell,r}^{1/2}\) gives
\begin{equation*}
(\bm K_{\ell,r}+\eta_\ell\bm I_{N_t})
\bm V_{\ell k}^{(t+1)}
=\bm B_{\ell k}^{(r)},
\end{equation*}
which is exactly the Spectral-NFP KKT system in \eqref{eq:spectral-kkt-system}. Therefore, maximizing the Spectral-NFP surrogate \(f_{s,r}\) is equivalent to applying Euclidean projected-gradient ascent to the transformed version of \(f_q\). Equivalently, Spectral-NFP is a \(\bm K_{\ell,r}\)-metric projected-gradient method for \(f_q\) in the original coordinates. At \(r=0\), \(\bm K_{\ell,0}=\lambda_{\ell,1}\bm I_{N_t}\), so the ellipsoid reduces to a Euclidean ball and the standard NFP projection is recovered.

\subsection{Nesterov Acceleration}\label{nesterov-acceleration}

Following Nesterov's extrapolation strategy \cite{Nesterov2018Lectures}, Spectral-NFP forms
\begin{equation}
\overline{\bm V}_{\ell k}^{(t)}
=\bm V_{\ell k}^{(t)}
+\theta_t\left(
\bm V_{\ell k}^{(t)}-\bm V_{\ell k}^{(t-1)}
\right).
\end{equation}
The auxiliary variables and curvature matrices are evaluated at the extrapolated beamformers, giving the gradient step
\begin{equation*}
\widehat{\bm V}_{\ell k}^{(t)}
=\overline{\bm V}_{\ell k}^{(t)}
+(\bm K_{\ell,r}^{(t)})^{-1}
\left(
\bm\Lambda_{\ell k}^{(t)}
-\bm L_\ell^{(t)}\overline{\bm V}_{\ell k}^{(t)}
\right).
\end{equation*}
The next iterate is obtained using the same power-constrained update as in Section IV-B, with the current beamformers replaced by the extrapolated ones. The extrapolated beamformers are not projected onto the power-feasible set beforehand.
The extrapolation coefficient can be selected as \cite{Nesterov2018Lectures}
\begin{equation}
\theta_t=
\max\left\{
\frac{t-2}{t+1},0
\right\},
\qquad t\ge1,
\end{equation}
with \(\underline{\bm V}^{(-1)}=\underline{\bm V}^{(0)}\).

\subsection{Computational Complexity}\label{beamformer-update-complexity}

We compare the method-specific complexity of the per-cell transmit update after the common auxiliary variables and the curvature matrix \(\bm L_\ell\) have been obtained. The following analysis focuses on the dominant dependence on the number of transmit antennas \(N_t\) and the retained rank \(r\), while treating the numbers of users and data streams as fixed.

\subsubsection{Spectral-NFP Complexity}
For \(1\le r\le N_t-1\), Spectral-NFP requires the leading \(r\) eigenvectors and the leading \(r+1\) eigenvalues of \(\bm L_\ell\), where the \((r+1)\)th eigenvalue determines the curvature assigned to the residual subspace.

Consider, for example, a block power or subspace-iteration method with a fixed number of iterations. For a general dense \(N_t\times N_t\) curvature matrix, multiplying the matrix by a block of \(r+1\) vectors has a leading-order complexity of \(\mathcal O(N_t^2r)\). The associated orthogonalization cost is \(\mathcal O(N_tr^2)\), which is lower order when \(r\ll N_t\).

After the required spectral components have been computed, the power-constrained beamformer update requires \(\mathcal O(KN_tdr)\) operations for the retained-subspace projections and beamformer reconstruction. The power multiplier is determined by scalar bisection. After the spectral energies have been computed, each bisection step evaluates a scalar power equation containing \(r+1\) terms and therefore costs \(\mathcal O(r+1)\), without requiring a matrix inversion. Treating the numbers of subspace iterations and bisection steps as fixed, the resulting complexity is
\begin{equation}
C_{\mathrm S}(r)
=\mathcal O\!\left(N_t^2r+KN_tdr\right).
\label{eq:spectral-update-complexity}
\end{equation}
For fixed \(K\) and \(d\), the leading-order complexity is \(\mathcal O(N_t^2r)\).

\subsubsection{WMMSE Complexity and Comparison}
The exact WMMSE/FP transmit update also determines the power multiplier by bisection. A direct implementation computes the inverse of a shifted \(N_t\times N_t\) curvature matrix at each bisection step and applies it to all beamformers, requiring \(\mathcal O(N_t^3+KN_t^2d)\) operations per step.

WMMSE-EVD avoids repeated matrix inversions by performing a full eigendecomposition of the curvature matrix before the bisection search. With the spectral energies precomputed, each bisection step evaluates a scalar power equation containing \(N_t\) terms. However, the full EVD still costs \(\mathcal O(N_t^3)\), and the beamformer projection and reconstruction cost \(\mathcal O(KN_t^2d)\). Thus, treating the number of bisection steps as fixed, both implementations have the leading-order complexity
\begin{equation}
C_{\mathrm W}
=\mathcal O\!\left(N_t^3+KN_t^2d\right)
\label{eq:wmmse-native-complexity}
\end{equation}
for fixed \(K\) and \(d\), although WMMSE-EVD reduces the computational cost by avoiding matrix inversion inside the bisection search.

With fixed user and stream counts, Spectral-NFP requires \(\mathcal O(N_t^2r)\) operations, compared with \(\mathcal O(N_t^3)\) for WMMSE. This reduces the leading-order cost when \(r\ll N_t\), while the WMMSE endpoint has the same cubic order.

\section{Theoretical Analysis}\label{theoretical-analysis}
\subsection{Certified Gain Relative to WMMSE}\label{certified-gain-relative-to-wmmse}

Proposition~\ref{prop:spectral-majorization} shows that increasing \(r\) monotonically tightens the spectral curvature upper bound. We next quantify how this tightening improves the one-step gain relative to the exact WMMSE transmit update.

Because \(\mathcal C\) and the fixed-auxiliary objectives separate across cells, we fix BS \(\ell\) and restrict \(\mathcal C\) to its cell-\(\ell\) factor throughout this subsection. We retain the cell index in all matrix quantities. For brevity, \(f_q(\underline{\bm V})\) and \(f_{s,r}(\underline{\bm V}\mid\underline{\bm V}^{(t)})\) denote the corresponding cell-wise transmit terms of the objectives defined in \eqref{eq:fp-quadratic-objective} and \eqref{eq:spectral-surrogate-objective}, with the auxiliary variables fixed. Let \(\underline{\bm V}^{(t)}=\{\bm V_{\ell k}^{(t)}\}_{k=1}^{K}\) be the current feasible beamformer block. Define the exact WMMSE transmit update as
\begin{equation}
\underline{\bm V}_W
=
\argmax_{\underline{\bm V}\in\mathcal C}
f_q(\underline{\bm V}).
\label{eq:wmmse-update-point}
\end{equation}
The resulting one-step gain is
\begin{equation}
\mathcal G_W
\triangleq
f_q(\underline{\bm V}_W)
-f_q(\underline{\bm V}^{(t)}).
\label{eq:wmmse-one-step-gain}
\end{equation}
Thus, \(\mathcal G_W\) measures the improvement in the fixed-auxiliary FP objective obtained by the exact WMMSE/FP transmit update from the current point.

Similarly, define the rank-\(r\) Spectral-NFP update at the same current point as
\begin{equation}
\underline{\bm V}_{S,r}
=
\argmax_{\underline{\bm V}\in\mathcal C}
f_{s,r}(\underline{\bm V}\mid\underline{\bm V}^{(t)}),
\label{eq:spectral-update-point}
\end{equation}
with the corresponding one-step surrogate gain
\begin{equation}
\mathcal G_{S,r}
\triangleq
f_{s,r}(\underline{\bm V}_{S,r}\mid\underline{\bm V}^{(t)})
-f_{s,r}(\underline{\bm V}^{(t)}\mid\underline{\bm V}^{(t)}).
\label{eq:spectral-one-step-gain}
\end{equation}
This quantity is the improvement certified by the rank-\(r\) Spectral-NFP lower-bound surrogate at the current point.
Since the surrogate lower-bounds \(f_q\) and is tight at
\(\underline{\bm V}^{(t)}\), it also certifies the actual fixed-auxiliary
improvement:
\begin{equation}
f_q(\underline{\bm V}_{S,r})
-f_q(\underline{\bm V}^{(t)})
\ge \mathcal G_{S,r}.
\label{eq:actual-gain-certificate}
\end{equation}
For \(\mathcal G_W>0\), \(\mathcal G_{S,r}/\mathcal G_W\) lower-bounds the fraction of the WMMSE transmit-objective improvement achieved by Spectral-NFP from the same current point.

\begin{proposition}[Rank-Monotone Certified Gain]\label{prop:rank-monotone-gain} For the same feasible current point, the same auxiliary variables, and exact rank-\(r\) spectral majorizers, define \(\mathcal G_N\triangleq\mathcal G_{S,0}\). Then
\begin{equation*}
0\le\mathcal G_N
=\mathcal G_{S,0}
\le\mathcal G_{S,1}
\le\cdots\le
\mathcal G_{S,N_t-1}
=\mathcal G_W.
\end{equation*}

\end{proposition}

\begin{IEEEproof}
For any feasible \(\underline{\bm V}\), substituting
\(\bm X=\bm V_{\ell k}\), \(\bm Z=\bm V_{\ell k}^{(t)}\),
\(\bm K=\bm K_{\ell,r}\), and \(\bm L=\bm L_\ell\) into
\eqref{eq:nonhomogeneous-bound}, and then using
\eqref{eq:spectral-linear-term}, gives
\begin{equation*}
\begin{aligned}
&f_{s,r}(\underline{\bm V}\mid\underline{\bm V}^{(t)})
-f_{s,r}(\underline{\bm V}^{(t)}\mid\underline{\bm V}^{(t)})\\
&=f_q(\underline{\bm V})-f_q(\underline{\bm V}^{(t)})
-\sum_k
\left\|\bm V_{\ell k}-\bm V_{\ell k}^{(t)}
\right\|_{\bm K_{\ell,r}-\bm L_\ell}^2.
\end{aligned}
\end{equation*}
Proposition~\ref{prop:spectral-majorization} implies \(\mathbf0\preceq\bm K_{\ell,r+1}-\bm L_\ell\preceq\bm K_{\ell,r}-\bm L_\ell\). Therefore, for every feasible \(\underline{\bm V}\),
\begin{multline}
f_{s,r+1}(\underline{\bm V}\mid\underline{\bm V}^{(t)})
-f_{s,r+1}(\underline{\bm V}^{(t)}\mid\underline{\bm V}^{(t)})\\
\ge f_{s,r}(\underline{\bm V}\mid\underline{\bm V}^{(t)})
-f_{s,r}(\underline{\bm V}^{(t)}\mid\underline{\bm V}^{(t)}).
\end{multline}

Maximizing both sides over \(\mathcal C\) proves the intermediate inequalities. The current point \(\underline{\bm V}^{(t)}\) is feasible and yields zero gain, proving nonnegativity. Finally, \(\bm K_{\ell,N_t-1}=\bm L_\ell\) makes \(f_{s,N_t-1}\) identical to \(f_q\) for the fixed auxiliary variables, so \(\mathcal G_{S,N_t-1}=\mathcal G_W\).
\end{IEEEproof}

\begin{theorem}[Relative Certified Gain]\label{thm:relative-certified-gain} Let \(\eta_{\ell,W}\ge0\) be the power-constraint multiplier associated with \(\underline{\bm V}_W\). Suppose that \(0\le r\le N_t-1\), \(\mathcal G_W>0\), and \(\lambda_{\ell,N_t}+\eta_{\ell,W}>0\). Then
\begin{equation}
\frac{\mathcal G_{S,r}}{\mathcal G_W}
\ge
\left[
1-
\frac{\lambda_{\ell,r+1}-\lambda_{\ell,N_t}}
{\lambda_{\ell,N_t}+\eta_{\ell,W}}
\right]_+.
\label{eq:relative-certified-gain}
\end{equation}
\end{theorem}

\begin{IEEEproof}
Define \(\bm\Delta_{W,\ell k}\triangleq\bm V_{W,\ell k}-\bm V_{\ell k}^{(t)}\). The KKT condition for \eqref{eq:wmmse-update-point} is
\begin{equation*}
(\bm L_\ell+\eta_{\ell,W}\bm I_{N_t})\bm V_{W,\ell k}
=\bm\Lambda_{\ell k}.
\end{equation*}
Using \(\bm V_{\ell k}^{(t)}=\bm V_{W,\ell k}-\bm\Delta_{W,\ell k}\) and
\begin{equation*}
2\operatorname{Re}\operatorname{tr}
(\bm\Delta_{W,\ell k}^H\bm V_{W,\ell k})
=\|\bm\Delta_{W,\ell k}\|_F^2
+\|\bm V_{W,\ell k}\|_F^2
-\|\bm V_{\ell k}^{(t)}\|_F^2,
\end{equation*}
the exact gain in \eqref{eq:wmmse-one-step-gain} becomes
\begin{equation*}
\begin{aligned}
\mathcal G_W
={}&\sum_k\|\bm\Delta_{W,\ell k}\|_{\bm L_\ell}^2
+\eta_{\ell,W}\sum_k\|\bm\Delta_{W,\ell k}\|_F^2\\
&+\eta_{\ell,W}\left(
\sum_k\|\bm V_{W,\ell k}\|_F^2
-\sum_k\|\bm V_{\ell k}^{(t)}\|_F^2
\right).
\end{aligned}
\end{equation*}
If \(\eta_{\ell,W}=0\), the multiplier terms vanish. If \(\eta_{\ell,W}>0\), complementary slackness and the feasibility of \(\underline{\bm V}^{(t)}\) give
\begin{equation*}
\sum_k\|\bm V_{W,\ell k}\|_F^2=P_\ell
\ge\sum_k\|\bm V_{\ell k}^{(t)}\|_F^2.
\end{equation*}
Since \(\bm L_\ell\succeq\lambda_{\ell,N_t}\bm I_{N_t}\), both cases imply
\begin{equation}
\mathcal G_W
\ge
(\lambda_{\ell,N_t}+\eta_{\ell,W})
\sum_k\|\bm\Delta_{W,\ell k}\|_F^2.
\label{eq:wmmse-gain-lower-bound}
\end{equation}

For any feasible \(\underline{\bm V}\), the same application of
\eqref{eq:nonhomogeneous-bound} gives
\begin{multline*}
f_{s,r}(\underline{\bm V}\mid\underline{\bm V}^{(t)})
-f_{s,r}(\underline{\bm V}^{(t)}\mid\underline{\bm V}^{(t)})\\
=f_q(\underline{\bm V})-f_q(\underline{\bm V}^{(t)})-\sum_k
\left\|\bm V_{\ell k}-\bm V_{\ell k}^{(t)}
\right\|_{\bm K_{\ell,r}-\bm L_\ell}^2.
\end{multline*}

The feasible point \(\underline{\bm V}_W\) is a valid candidate in \eqref{eq:spectral-update-point}, so
\begin{equation*}
\begin{aligned}
\mathcal G_{S,r}
&\ge
\mathcal G_W
-\sum_k\operatorname{tr}
[\bm\Delta_{W,\ell k}^H(\bm K_{\ell,r}-\bm L_\ell)\bm\Delta_{W,\ell k}]\\
&\ge
\mathcal G_W
-\|\bm K_{\ell,r}-\bm L_\ell\|_2
\sum_k\|\bm\Delta_{W,\ell k}\|_F^2.
\end{aligned}
\end{equation*}

The eigenvalues of \(\bm K_{\ell,r}-\bm L_\ell\) are zero in the retained subspace and \(\lambda_{\ell,r+1}-\lambda_{\ell,a}\) in the residual subspace. Therefore,
\begin{equation*}
\|\bm K_{\ell,r}-\bm L_\ell\|_2
=\lambda_{\ell,r+1}-\lambda_{\ell,N_t}.
\end{equation*}

Combining this identity with \eqref{eq:wmmse-gain-lower-bound} gives
\begin{equation*}
\mathcal G_{S,r}
\ge
\mathcal G_W
\left(
1-
\frac{\lambda_{\ell,r+1}-\lambda_{\ell,N_t}}
{\lambda_{\ell,N_t}+\eta_{\ell,W}}
\right).
\end{equation*}

Moreover, \(\mathcal G_{S,r}\ge0\) because \(\underline{\bm V}^{(t)}\) is feasible in \eqref{eq:spectral-update-point} and yields zero surrogate gain. Taking the larger of the two lower bounds proves \eqref{eq:relative-certified-gain}.
\end{IEEEproof}

Theorem~\ref{thm:relative-certified-gain} converts the relative
surrogate-certified one-step improvement of Spectral-NFP into a property of
the curvature spectrum. In particular, the relevant quantity is the
normalized residual spectral spread
\begin{equation}
\frac{\lambda_{\ell,r+1}-\lambda_{\ell,N_t}}
{\lambda_{\ell,N_t}+\eta_{\ell,W}}.
\end{equation}
As \(r\) increases, \(\lambda_{\ell,r+1}\) is nonincreasing, so this residual penalty decreases and the gain-ratio lower bound increases monotonically toward one. This agrees with Proposition~\ref{prop:rank-monotone-gain}, which orders the optimized surrogate gains directly. Thus, Spectral-NFP approaches the WMMSE transmit-update gain whenever the discarded curvature eigenvalues are sufficiently concentrated near the smallest eigenvalue. A small value of \(\lambda_{\ell,r+1}\) alone is not sufficient; the residual spread must be small relative to the minimum regularized WMMSE curvature.

Theorem~\ref{thm:relative-certified-gain} expresses the guaranteed gain ratio through the curvature spectrum and the WMMSE power multiplier. We next study its finite-dimensional probability bound under a random-curvature model and derive a rank-selection rule in the large-system limit.

\subsection{Finite-Dimensional Wishart Certificate}\label{finite-dimensional-wishart-certificate}

To analyze the spectrum of \(\bm L_\ell\), define each weighted effective-channel block and horizontally concatenate all \(LK\) blocks as
\begin{equation}
\begin{aligned}
\bm C_{ij,\ell}
&\triangleq
\sqrt{w_{ij}}
\bm H_{ij,\ell}^H
\bm Y_{ij}
(\bm I_d+\bm\Gamma_{ij})^{1/2}
\in\mathbb C^{N_t\times d},\\
\bm C_\ell
&\triangleq
\left[\bm C_{ij,\ell}\right]_{i,j}
\in\mathbb C^{N_t\times M},
\qquad M=LKd.
\end{aligned}
\label{eq:wishart-factor}
\end{equation}
Here, \((\bm I_d+\bm\Gamma_{ij})^{1/2}\) is the principal
Hermitian square root. The curvature matrix can then be written as the Gram
matrix
\begin{equation}
\bm L_\ell=\bm C_\ell\bm C_\ell^H.
\label{eq:wishart-factorization}
\end{equation}

For the following analysis, we model the entries of the effective-channel factor \(\bm C_\ell\) as independent \(\mathcal{CN}(0,1)\) variables. The resulting curvature matrix \(\bm L_\ell=\bm C_\ell\bm C_\ell^H\) is complex Wishart, with ordered-eigenvalue density given by the following lemma.

\begin{Lemma}[Ordered Complex-Wishart Eigenvalues~\cite{TulinoVerdu2004RMT}]
\label{lem:ordered-complex-wishart}
Let \(\bm C\in\mathbb C^{N_t\times M}\), with \(M\ge N_t\), have independent \(\mathcal{CN}(0,1)\) entries, and let
\begin{equation*}
\lambda_1\ge\lambda_2\ge\cdots\ge\lambda_{N_t}\ge0
\end{equation*}
be the ordered eigenvalues of \(\bm C\bm C^H\). Their joint probability density is
\begin{equation*}
p_{M,N_t}^{\downarrow}(\bm\lambda)
={}
K_{M,N_t}^{-1}
\prod_{a=1}^{N_t}e^{-\lambda_a}\lambda_a^{M-N_t}
\prod_{a<b}(\lambda_a-\lambda_b)^2
\end{equation*}
for \(\bm\lambda\in\mathcal D_{N_t}
\triangleq\{\bm\lambda:\lambda_1\ge\cdots\ge
\lambda_{N_t}\ge0\}\), and zero otherwise, where \(K_{M,N_t}\) is the corresponding normalizing constant.
\end{Lemma}

For a prescribed relative loss \(\epsilon\in(0,1)\), we seek the probability that Spectral-NFP achieves at least a fraction \(1-\epsilon\) of the WMMSE transmit-update gain:
\begin{equation*}
\Pr\!\left(
\frac{\mathcal G_{S,r}}{\mathcal G_W}\ge1-\epsilon
\right).
\end{equation*}
By Theorem~\ref{thm:relative-certified-gain} and \(\eta_{\ell,W}\ge0\), we obtain
\begin{equation*}
\Pr\!\left(
\frac{\mathcal G_{S,r}}{\mathcal G_W}\ge1-\epsilon
\right)
\ge
\Pr\!\left(
\lambda_{r+1}\le(1+\epsilon)\lambda_{N_t}
\right).
\end{equation*}

Under the idealized Wishart model, Lemma~\ref{lem:ordered-complex-wishart} gives the probability on the right-hand side. Define its integration region as
\begin{equation*}
\mathcal A_{r,\epsilon}
=\left\{
\bm\lambda\in\mathcal D_{N_t}:
\lambda_{r+1}\le(1+\epsilon)\lambda_{N_t}
\right\}.
\end{equation*}
Integrating the ordered-eigenvalue density over this region gives
\begin{multline*}
\Pr(\mathcal A_{r,\epsilon})
=
K_{M,N_t}^{-1}
\int_{\mathcal A_{r,\epsilon}}
\prod_{a=1}^{N_t}e^{-\lambda_a}\lambda_a^{M-N_t}\\
\times
\prod_{a<b}(\lambda_a-\lambda_b)^2
\,d\lambda_1\cdots d\lambda_{N_t}.
\end{multline*}

Combining the deterministic gain bound with Lemma~\ref{lem:ordered-complex-wishart} yields the following result directly.

\begin{proposition}[Finite-Dimensional Wishart Certificate]\label{prop:finite-wishart-certificate} Under the i.i.d.-Gaussian factor model above with \(M\ge N_t\), let \(0\le r\le N_t-1\), \(\epsilon\in(0,1)\), and assume \(\mathcal G_W>0\) almost surely. Then
\begin{equation*}
\Pr\!\left(
\frac{\mathcal G_{S,r}}{\mathcal G_W}\ge1-\epsilon
\right)
\ge
\Pr(\mathcal A_{r,\epsilon}).
\end{equation*}
\end{proposition}

Because the sufficient spectral condition removes the nonnegative, sample-dependent multiplier \(\eta_{\ell,W}\), \(\Pr(\mathcal A_{r,\epsilon})\) is a lower bound on the probability of attaining the target gain ratio. Since these events are nested in \(r\), for a tolerated failure probability \(\delta\in(0,1)\), one may choose the smallest \(r\) satisfying \(\Pr(\mathcal A_{r,\epsilon})\ge 1-\delta\).
However, the ordered \(N_t\)-dimensional integral generally does not admit a simple closed-form expression and becomes difficult to evaluate as \(N_t\) grows. We therefore turn to the large-system limit to obtain a tractable rank-selection rule.

\subsection{Asymptotic Rank-Selection Rule}\label{asymptotic-rank-selection-rule}

To obtain an explicit rank-selection rule, we consider the large-system limit in which both \(N_t\) and the factor count \(M=LKd\) grow with a fixed ratio.

To simplify notation in this subsection, we suppress the cell index and write
\(\bm C=\bm C_\ell\), \(\bm L=\bm L_\ell\),
\(\lambda_a=\lambda_{\ell,a}\), and
\(\eta_W=\eta_{\ell,W}\).

Let
\begin{equation}
N_t,M\rightarrow\infty,
\qquad
\frac{N_t}{M}\rightarrow c\in(0,1),
\label{eq:large-system-scaling}
\end{equation}
and define the empirical distribution of the normalized eigenvalues:
\begin{equation}
F_{N_t}(x)
=\frac{1}{N_t}
\sum_{a=1}^{N_t}
\bm 1\left\{
\frac{\lambda_a(\bm L)}{M}\le x
\right\}.
\label{eq:empirical-spectrum}
\end{equation}
Thus, \(F_{N_t}(x)\) is the fraction of transmit-side curvature directions whose normalized curvature does not exceed \(x\).

\begin{Lemma}[Marchenko--Pastur Law~\cite{MarchenkoPastur1967,TulinoVerdu2004RMT}]
\label{lem:mp-law}
Under the i.i.d.-Gaussian factor model and the scaling in \eqref{eq:large-system-scaling}, the empirical distribution \(F_{N_t}\) in \eqref{eq:empirical-spectrum} converges almost surely to a deterministic distribution \(F_c\) with density
\begin{equation}
f_c(x)
=
\begin{cases}
\dfrac{\sqrt{(b_c-x)(x-a_c)}}{2\pi cx},
& a_c\le x\le b_c,\\[4pt]
0, & \text{otherwise},
\end{cases}
\label{eq:mp-density}
\end{equation}
where
\begin{equation}
a_c=(1-\sqrt c)^2,
\qquad
b_c=(1+\sqrt c)^2.
\label{eq:mp-edges}
\end{equation}
\end{Lemma}

Let the retained rank \(r\) increase with \(N_t\) such that
\begin{equation}
\frac{r}{N_t}\longrightarrow\rho,
\qquad \rho\in(0,1).
\label{eq:rank-scaling}
\end{equation}

Thus, \(r\) varies with the antenna dimension rather than remaining fixed.
Under the descending ordering \(\lambda_1\ge\cdots\ge\lambda_{N_t}\), the
eigenvalue \(\lambda_{r+1}\) has a fraction \((N_t-r)/N_t\) of the
eigenvalues at or below it. Since \(F_c\) is continuous and strictly increasing
on \((a_c,b_c)\), the quantile consequence of Lemma~\ref{lem:mp-law} gives
\begin{equation}
\frac{\lambda_{r+1}}{M}
\xrightarrow{\mathrm{a.s.}}
F_c^{-1}(1-\rho),
\label{eq:discarded-eigenvalue-limit}
\end{equation}
This quantile argument applies only when \(r/N_t\) remains away from the two spectral endpoints and therefore does not determine the lower endpoint \(\lambda_{N_t}\). The following lower-edge result supplies this missing limit.

\begin{Lemma}[Bai--Yin Lower-Edge Limit~\cite{BaiYin1993SmallestEigenvalue}]
\label{lem:bai-yin-lower-edge}
Under the i.i.d.-Gaussian factor model and the scaling in \eqref{eq:large-system-scaling}, the smallest eigenvalue satisfies
\begin{equation}
\frac{\lambda_{N_t}}{M}
\xrightarrow{\mathrm{a.s.}}
a_c.
\label{eq:smallest-eigenvalue-limit}
\end{equation}
\end{Lemma}

Combining \eqref{eq:discarded-eigenvalue-limit} and
\eqref{eq:smallest-eigenvalue-limit} with
Theorem~\ref{thm:relative-certified-gain} yields the following asymptotic
result.

\begin{proposition}[Wishart Rank-Fraction Certificate]
\label{prop:wishart-rank-fraction}
Under the i.i.d. Gaussian factor model above, the large-system scaling in
\eqref{eq:large-system-scaling}, the rank scaling in \eqref{eq:rank-scaling},
and the condition \(\mathcal G_W>0\) almost surely along the sequence,
\begin{equation}
\liminf_{N_t\rightarrow\infty}
\frac{\mathcal G_{S,r}}{\mathcal G_W}
\ge
\left[
1-
\frac{F_c^{-1}(1-\rho)-a_c}{a_c}
\right]_+
\quad\mathrm{a.s.}
\label{eq:asymptotic-gain-bound}
\end{equation}
\end{proposition}

\begin{IEEEproof}
Since \(\eta_W\ge0\), \eqref{eq:relative-certified-gain} gives
\begin{equation}
\frac{\mathcal G_{S,r}}{\mathcal G_W}
\ge
\left[
1-
\frac{\lambda_{r+1}-\lambda_{N_t}}
{\lambda_{N_t}}
\right]_+.
\label{eq:eta-free-gain-bound}
\end{equation}
On the probability-one event where \eqref{eq:discarded-eigenvalue-limit} and \eqref{eq:smallest-eigenvalue-limit} hold,
\begin{equation*}
\begin{aligned}
\frac{\lambda_{r+1}-\lambda_{N_t}}{\lambda_{N_t}}
&=
\frac{\lambda_{r+1}/M-\lambda_{N_t}/M}
{\lambda_{N_t}/M}\\
&\xrightarrow{\mathrm{a.s.}}
\frac{F_c^{-1}(1-\rho)-a_c}{a_c}.
\end{aligned}
\end{equation*}
Since \([1-x]_+\) is continuous, applying this limit to \eqref{eq:eta-free-gain-bound} gives \eqref{eq:asymptotic-gain-bound}. At the support boundaries, \(F_c^{-1}\) is understood as the generalized inverse.
\end{IEEEproof}

Proposition~\ref{prop:wishart-rank-fraction} relates the asymptotic gain ratio to the retained rank fraction. It therefore yields the following rank-selection criterion for achieving a prescribed gain ratio.

\begin{corollary}[Asymptotic Rank-Selection Rule]
\label{cor:asymptotic-rank-selection}
Under the conditions of Proposition~\ref{prop:wishart-rank-fraction}, for a target loss \(\epsilon\in(0,1)\), the rank condition
\begin{equation}
\rho
\ge
1-F_c\!\left((1+\epsilon)a_c\right)
\label{eq:mp-rank-fraction-rule}
\end{equation}
implies
\begin{equation}
\liminf_{N_t\rightarrow\infty}
\frac{\mathcal G_{S,r}}{\mathcal G_W}
\ge1-\epsilon
\quad\mathrm{a.s.}
\end{equation}
\end{corollary}

\begin{IEEEproof}
The condition in \eqref{eq:mp-rank-fraction-rule} and the monotonicity of \(F_c\) imply that the right-hand side of \eqref{eq:asymptotic-gain-bound} is at least \(1-\epsilon\). The result then follows directly from Proposition~\ref{prop:wishart-rank-fraction}.
\end{IEEEproof}

We next give a numerical example to illustrate how to select the rank for a prescribed gain ratio. Let \(c=0.001\) and set the target loss to \(\epsilon=0.1\). Numerical evaluation of the Marchenko--Pastur CDF gives
\begin{equation*}
a_c\approx0.937754,\qquad
F_c^{-1}(0.8)\approx1.031095,
\end{equation*}
so \(r/N_t=0.2\) makes the Marchenko--Pastur expression in \eqref{eq:asymptotic-gain-bound} approximately \(0.90046\). Equivalently, \eqref{eq:mp-rank-fraction-rule} requires \(r/N_t\) to approach a value no smaller than \(0.19627\) for a \(90\%\) asymptotic certificate. Thus, this idealized example predicts that retaining roughly \(20\%\) of the transmit dimensions is sufficient to capture at least \(90\%\) of the full-curvature certified gain.

\begin{figure*}[t]
\centering
\subfigure[Outer iteration]{\includegraphics[width=0.45\linewidth]{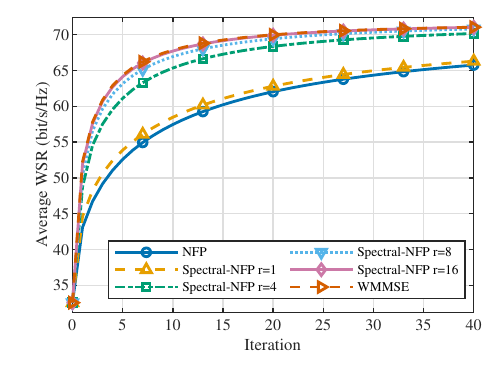}}\hfill
\subfigure[Cumulative update time]{\includegraphics[width=0.45\linewidth]{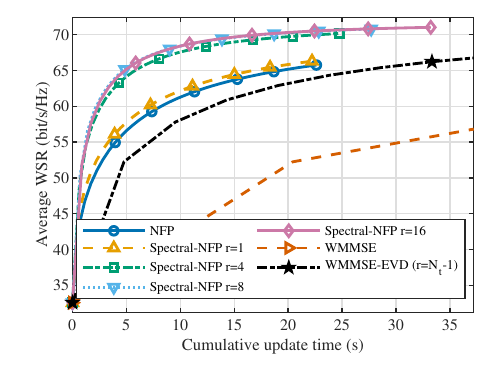}}
\caption{Average WSR on the geometry-based shadowing channels with $N_t=256$ and six users per cell. The time axis focuses on NFP and low-rank Spectral-NFP while retaining the early WMMSE and WMMSE-EVD trajectories.}
\label{fig:geometry-iteration}
\end{figure*}

\begin{figure*}[t]
\centering
\subfigure[Accuracy scaling]{\includegraphics[width=0.45\linewidth]{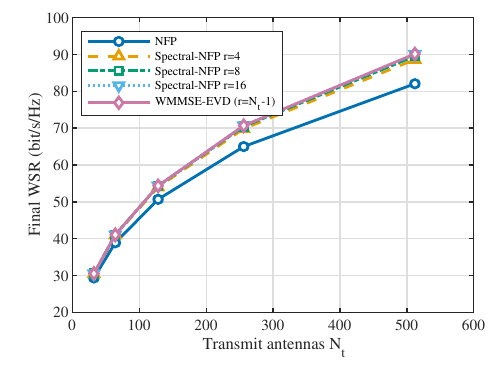}}\hfill
\subfigure[Runtime scaling]{\includegraphics[width=0.45\linewidth]{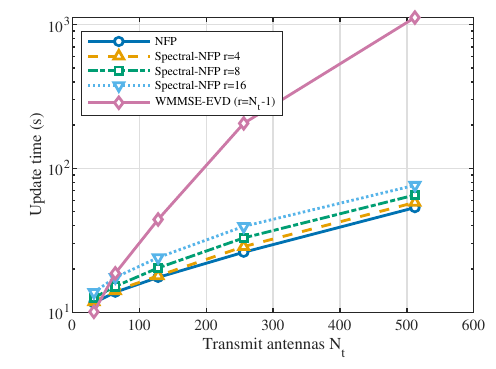}}
\caption{Antenna scaling with six users per cell. (a) Final WSR. (b) Median low-rank runtime and single-run WMMSE-EVD runtime.}
\label{fig:antenna-scaling}
\end{figure*}

\begin{figure*}[t]
\centering
\subfigure[Per-user WSR]{\includegraphics[width=0.45\linewidth]{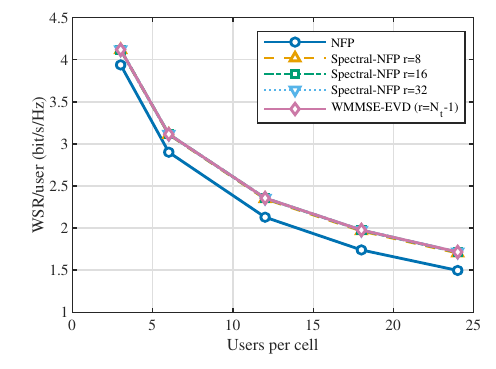}}\hfill
\subfigure[Runtime]{\includegraphics[width=0.45\linewidth]{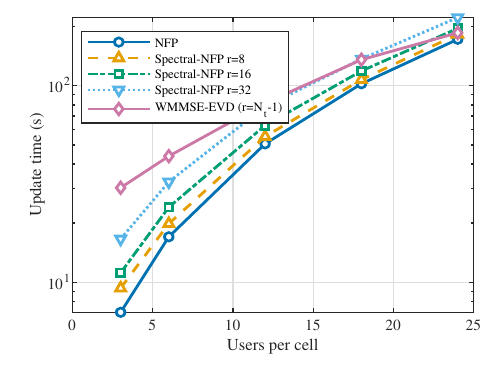}}
\caption{User-load scaling at $N_t=128$. (a) Final WSR per user. (b) Median NFP and Spectral-NFP runtime and single-run WMMSE-EVD runtime.}
\label{fig:user-load-scaling}
\end{figure*}

The required rank fraction depends on the dimension ratio and the target gain ratio. In this example, \(c=0.001\) corresponds to \(M=1000N_t\), so the Wishart spectrum is strongly concentrated. Section~\ref{numerical-results} examines the performance of low-rank Spectral-NFP on geometry-based multicell channels.

\section{Numerical Results}\label{numerical-results}

This section compares Spectral-NFP with NFP and WMMSE in terms of WSR and running time, antenna and user-load scaling, SNR robustness, and the theoretical gain bound. WMMSE-EVD implements the same WMMSE transmit update using a full EVD followed by scalar bisection in the eigenbasis, avoiding repeated matrix inversions.

\subsection{Experimental Setup and Channel Scenarios}\label{experimental-setup-and-channel-scenarios}

The main experiments consider a seven-cell wrapped-around hexagonal network following the simulation model in~\cite{ShenEtAl2024AcceleratingQT}. Each cell is divided into three sectors, and adjacent BSs are separated by 800 m. Users are independently placed in their serving cells subject to a minimum BS--user distance of 300 m. The wrap-around construction provides every cell with the same interference geometry and avoids privileging the center cell. The distance-dependent path loss is
\begin{equation}
\operatorname{PL}(q_{\ell k,i})
=15.3+37.6\log_{10}(q_{\ell k,i})+\xi\quad\text{dB},
\end{equation}
where \(q_{\ell k,i}\) is the distance in meters from BS \(i\) to user \((\ell,k)\), and \(\xi\sim\mathcal N(0,8^2)\) models shadowing. The small-scale coefficients are modeled with independent phases and are scaled by the corresponding large-scale attenuation. Every user has \(N_r=4\) receive antennas. The main convergence comparison uses \(N_t=256\) and six users per cell. The scaling experiments vary \(N_t\in\{32,64,128,256,512\}\), the user load from 3 to 24 users per cell, and the SNR from 70 to 120 dB. Within each scaling sweep, nested antenna or user subsets are used so that the compared dimensions share the same underlying propagation realizations.

Each reported setting averages the WSR over 100 independent channel realizations. Unless otherwise stated, each user is assigned two data streams, all user weights are equal, the per-BS power budget is 100, and the noise variance is \(10^{-8}\). All algorithms start from the same RZF beamformers and run for 40 outer iterations with a minibatch size of five. We report median running times over three runs for NFP and low-rank Spectral-NFP, and single-run times for the WMMSE baselines.

\subsection{Main Geometry-Based Multicell Results}\label{main-geometry-based-multicell-results}

Fig.~\ref{fig:geometry-iteration}(a) compares the WSR trajectories at \(N_t=256\). After five outer iterations, NFP attains 81.744\% of the WMMSE WSR, whereas Spectral-NFP with \(r=1,4,8,16\) attains 83.742\%, 94.999\%, 98.188\%, and 99.698\%, respectively. These ranks retain only 0.391\%, 1.563\%, 3.125\%, and 6.25\% of the transmit dimension.

Fig.~\ref{fig:geometry-iteration}(b) shows the corresponding time comparison. The cumulative times for five iterations are 11.9\% of the WMMSE-EVD time for NFP and 11.7\%, 13.0\%, 14.6\%, and 17.6\% for the four Spectral-NFP ranks. Thus, low-rank updates substantially improve early WSR with only a modest additional cost over NFP.

\subsection{Scaling With Antenna Dimension and User Load}\label{scaling-with-antenna-dimension-and-user-load}

To study the effect of the transmit antenna dimension, Fig.~\ref{fig:antenna-scaling}(a) compares the final WSR as \(N_t\) grows from 32 to 512. Across this range, \(r=4\), \(r=8\), and \(r=16\) achieve at least 98.326\%, 99.354\%, and 99.921\%, respectively, of the WMMSE WSR. At \(N_t=512\), ranks 8 and 16 retain only 1.563\% and 3.125\% of the transmit dimension while attaining 99.354\% and 99.921\% of the WMMSE WSR. Thus, fixed low ranks remain close to WMMSE as the antenna dimension increases.

Fig.~\ref{fig:antenna-scaling}(b) compares the corresponding running times. Full EVD is faster at \(N_t=32\), and its cost is comparable to the low-rank updates at \(N_t=64\). From \(N_t=128\) onward, however, the WMMSE-EVD time increases much faster with the antenna dimension. At \(N_t=512\), ranks 8 and 16 use only 5.84\% and 6.83\% of the WMMSE-EVD time while attaining 99.354\% and 99.921\% of the WMMSE final WSR, respectively. The computational advantage of the low-rank updates therefore becomes pronounced for larger antenna arrays.

\begin{figure*}[!t]
\centering
\begin{minipage}[t]{0.48\textwidth}
\centering
\includegraphics[width=0.90\linewidth]{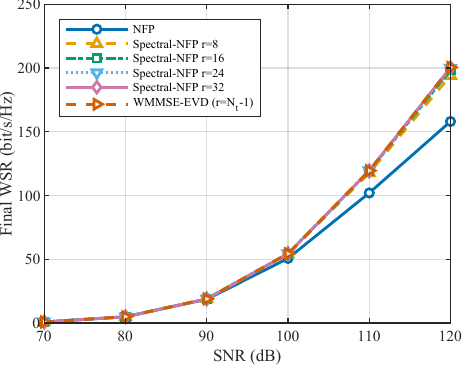}
\caption{Final WSR versus SNR at $N_t=128$ with six users per cell.}
\label{fig:snr-robustness}
\end{minipage}\hfill
\begin{minipage}[t]{0.48\textwidth}
\centering
\includegraphics[width=0.90\linewidth]{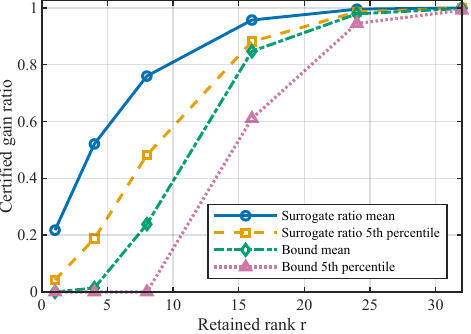}
\caption{Mean and empirical 5th percentile of the surrogate-gain ratio \(\mathcal G_{S,r}/\mathcal G_W\) and its theoretical lower bound, evaluated at common current beamformers.}
\label{fig:certificate}
\end{minipage}
\end{figure*}

Fig.~\ref{fig:user-load-scaling}(a) varies the number of users per cell at fixed \(N_t=128\) and reports the final WSR per user after 40 iterations. The WSR per user decreases as more users share the same cell power. More importantly, the final NFP WSR decreases from 95.626\% to 87.228\% of the WMMSE-EVD WSR as the number of users per cell increases from 3 to 24, whereas Spectral-NFP with \(r=8\) remains between 99.275\% and 99.978\%. Spectral-NFP with \(r=16\) remains above 99.857\% throughout the sweep. Thus, the advantage of retaining directional curvature becomes more pronounced as the user load increases.

Fig.~\ref{fig:user-load-scaling}(b) shows how the running time changes with the number of users per cell. Increasing the number of users raises the partial spectral and beamformer-reconstruction work, whereas the dimension of the WMMSE-EVD factorization remains fixed. Low-rank Spectral-NFP is faster at light and moderate loads. At 24 users per cell, Spectral-NFP with \(r=16\) has a runtime comparable to WMMSE-EVD, with a measured time ratio of 1.05.

\subsection{Robustness Across Operating SNR}\label{robustness-across-operating-snr}

We next fix \(N_t=128\) and six users per cell and vary the per-BS power as \(P=\sigma^2 10^{\mathrm{SNR}/10}\) for SNR values from 70 to 120 dB. Here, SNR is defined as the ratio of the per-BS transmit power to the noise variance, expressed in dB; it does not include path loss or interference.

Fig.~\ref{fig:snr-robustness} reports the final WSR after 40 iterations. At 70--90 dB, NFP already attains between 98.253\% and 99.974\% of the WMMSE-EVD WSR, and Spectral-NFP with \(r=8\) attains at least 99.998\%. The separation becomes clearer as the transmit power increases. At 100 dB, NFP and Spectral-NFP with \(r=8\) attain 93.229\% and 99.895\% of the WMMSE-EVD WSR, respectively. At 120 dB, the corresponding ratios are 78.830\% for NFP and 96.820\%, 98.974\%, 99.707\%, and 99.946\% for Spectral-NFP with \(r=8,16,24,\) and 32. Thus, more curvature directions become useful at higher transmit power, while moderate ranks remain close to WMMSE-EVD throughout the tested range.

\subsection{Certified-Gain Validation}\label{certified-gain-validation}

We evaluate the gain ratio and the lower bound in Theorem~\ref{thm:relative-certified-gain} at common current beamformers for the \(N_t=128\) scenario. The evaluation uses iterations \(0,1,5,10,20,\) and 40, with the WMMSE update computed by EVD.

Fig.~\ref{fig:certificate} reports the optimized-surrogate gain ratio
\(\mathcal G_{S,r}/\mathcal G_W\) and the eta-aware deterministic lower bound
in \eqref{eq:relative-certified-gain}. At \(r=8\), the mean and empirical 5th
percentile of this surrogate-certified ratio are 0.7599 and 0.4820. They
increase to 0.9571 and 0.8812 at \(r=16\), and to 0.9953 and 0.9860 at
\(r=24\). The corresponding mean lower bounds are 0.2385, 0.8469, and 0.9788.
The observed surrogate-gain ratios satisfy the theoretical lower bound at all evaluated points, within numerical precision.

\section{Conclusion}\label{conclusion}

This paper proposed Spectral-NFP as a rank-controlled refinement of NFP for multicell MIMO WSR maximization. Spectral-NFP preserves the leading eigenpairs of the WMMSE/FP transmit curvature and upper-bounds only the remaining eigenvalues. The retained rank therefore controls the transition from the scaled-identity NFP update at \(r=0\) to the exact WMMSE/FP transmit update at \(r=N_t-1\). The resulting low-rank structure avoids a full curvature decomposition while retaining a closed-form beamformer reconstruction and a scalar power-multiplier search.

We derived a deterministic lower bound on the ratio between the one-step transmit-objective gains of Spectral-NFP and WMMSE. This bound is determined by the curvature spectrum and increases with the retained rank. Under an idealized Wishart curvature model, we further obtained a finite-dimensional probability guarantee and a large-system rank-selection rule. We also showed that, for a fixed surrogate curvature matrix, Spectral-NFP is equivalent to Euclidean projected-gradient ascent after a linear coordinate transformation, which motivates the Nesterov-like extension.

Across the tested antenna dimensions from 32 to 512, ranks 8 and 16 attain at least 99.354\% and 99.921\% of the WMMSE final WSR, respectively. At \(N_t=512\), their measured update times are only 5.84\% and 6.83\% of the WMMSE-EVD update time. The experiments further show that higher user loads and transmit powers favor retaining more curvature directions.

Future work will explore adaptive rank selection based on the curvature spectrum.

\bibliographystyle{IEEEtran}
\bibliography{IEEEabrv,references}

\end{document}